\documentclass[10pt]{amsart}
\usepackage[english]{babel}
\usepackage{amsmath,amssymb, amsfonts, amsthm, amssymb, amscd, graphicx,enumerate, 
calc,ifthen,latexsym,alltt,verbatim, textcomp}

\newtheorem{theorem}{Theorem}[section]
\newtheorem{theorem*}{Theorem}
\newtheorem{lemma}[theorem]{Lemma}
\newtheorem{proposition}[theorem]{Proposition}
\newtheorem{proposition*}[theorem*]{Proposition}
\newtheorem{corollary}[theorem]{Corollary}
\newtheorem{corollary*}[theorem*]{Corollary}
\newtheorem{definition}[theorem]{Definition}
\newtheorem{definition*}[theorem*]{Definition}
\newtheorem{example}[theorem]{Example}
\newtheorem{example*}[theorem*]{Example}
\newtheorem{note}[theorem]{Note}
\newtheorem{note*}[theorem*]{Note}
\newtheorem{remark}[theorem]{Remark}
\newtheorem{remark*}[theorem*]{Remark}

\begin{document}

\title{Selfish Mining in Ethereum}

\subjclass[2010]{68M01, 60G40, 91A60.}
\keywords{Bitcoin, Ethereum, blockchain, proof-of-work, selfish mining, Catalan numbers, Dyck path, random walk.}

\author{Cyril Grunspan}
\address{L{\'e}onard de Vinci, P{\^o}le Univ., Research Center, Paris-La D{\'e}fense, Labex R{\'e}fi, France}
\email{cyril.grunspan@devinci.fr}

\author{Ricardo P\'erez-Marco}
\address{CNRS, IMJ-PRG, Labex R{\'e}fi, Paris, France}
\email{ricardo.perez.marco@gmail.com}

\address{\tiny Author's Bitcoin Beer Address (ABBA)\footnote{\tiny Send some bitcoins to support our research at the pub.}:
1KrqVxqQFyUY9WuWcR5EHGVvhCS841LPLn} 

\address{\includegraphics[scale=0.33]{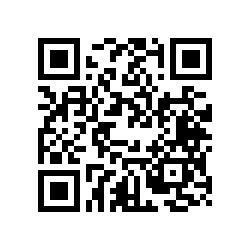}}

\date{April 30th, 2019}

\begin{abstract}
 We study selfish mining in Ethereum. The problem is combinatorially more complex than in 
  Bitcoin because of major differences in the reward system and a different difficulty adjustment formula. 
  Equivalent strategies in Bitcoin do have different profitabilities in Ethereum. 
  The attacker can either broadcast his 
  fork one block by one, or keep them secret as long as possible and 
  publish them all at once at the end of an attack cycle. The first strategy 
  is damaging for substantial hashrates, and we show that the second strategy is even worse. This
  confirms what we already proved for Bitcoin: Selfish mining is most of all  
  an attack on the difficulty adjustment formula. 
  We show that the current reward for signaling uncle 
  blocks is a weak incentive for the attacker to signal blocks. We compute the profitabilities of different strategies and find out 
  that for a large parameter space values, strategies that do not signal blocks are the best ones. We compute closed-form formulas for the 
  apparent hashrates for these strategies and compare them. We use a direct combinatorics analysis with Dyck words to 
  find these closed-form formulas.
\end{abstract}

\maketitle

\section{Introduction}

\subsection{Selfish mining strategies in Ethereum}

Research on selfish mining (in short SM) in Ethereum is quite recent. We can mention as recent 
contributions \cite{RZ18} (numerical study) and {\cite{NF19}}.

The authors of {\cite{NF19}} use a Markov chain model and compute the 
stationary probability. Then they study what they call the ``absolute revenue'' 
of the attacker which corresponds to the apparent hashrate after a difficulty adjustment 
as explained in our articles on blockwithholding attacks in the Bitcoin network 
(see \cite{GPM18}, \cite{GPM18e}, \cite{GPM19a}). 
Their theoretical analysis seems also confirmed by their numerical 
simulations. They do not provide
closed-form formulas (for example Formulas (8) and (9) in Section 3-E involve double infinite
sums). But more importantly, their study is limited to the following strategy of the attacker:

\medskip

\begin{enumerate}

\item \label{point1}The attacker refers to all possible orphan blocks;
  
  \item \label{point2} When new blocks are validated by the honest
  miners, the attacker makes public the part of his fork sharing the
  same height as the ``honest'' blockchain.
\end{enumerate}
\medskip

(See Algorithm 1 in \cite{NF19}, Lines 1 and 19 from Section 3-C)

\medskip

We label this strategy as ``Strategy 1'' or SM1. The procedure of a Bitcoin selfish miner to 
release his secret fork is irrelevant for the profitability of the classical selfish mining 
attack. However, this is not so in Ethereum. In particular, 
the precise algorithm presented in \cite{NF19} is not the most profitable as we will prove. 
An alternative strategy for the attacker would be to keep secret all his fork until
he is on the edge of being caught-up by the honest miners. Then, and only at this critical moment,
he would release his complete fork and override the
public blockchain. We label this second strategy as ``Strategy 2'' or SM2. In Bitcoin, both
strategies have the same effect since only matters 
the number of blocks mined by the attacker and added to the official
blockchain. But in Ethereum, this is not so because of the 
different reward incentives that gives rewards to ``nephew'' blocks  who refer to ``uncle''
blocks. ``Uncle'' blocks are orphan blocks with a parent in the official blockchain, and the descendants 
of this parent in the official blockchain are its ``nephew'' blocks. Also uncle blocks get rewards 
when referred by nephews.

\medskip

\subsection{Performance of Ethereum selfish mining strategies.}
To understand what the best strategy for the attacker is, we need an in-deep 
knowledge of the nature of the selfish mining attack. In \cite{GPM18} we give a correct economic modeling  
with a model of repetition game, and we consider the time element that is absent from older Markov chain models. 
What is important for the attacker is to maximize the number of validated blocks in the official blockchain 
\textit{per unit of time}, which is different from the percentage of blocks he validates. 
With this correct modeling, it becomes then clear that the attack is an exploit on Bitcoin's difficulty adjustment formula, 
that does include the orphan blocks. Then the attacker lowers artificially the difficulty, at the expense of orphaned honest blocks,
and succeeds to validate more blocks per unit of time.

\medskip

Point (\ref{point2}) in ``Strategy 1'' creates 
numerous competitions between the attacker's fork and
the honest blockchain. This increases the production of orphan blocks that
becomes important for a substantial hashrate of the attacker. Signaling these
orphan blocks yields additional rewards to the attacker, but it goes against its main goal
to lower the difficulty. Indeed, the difficulty's adjustment formula in Ethereum counts for ``uncles'', that 
are the orphan blocks directly attached to the main chain. 
Therefore, increasing the number of uncles by Point \ref{point2} has the following
contradictory effects: On one hand, the attacker's revenue increases
because of the new ``inclusion rewards'', but on the other hand, the difficulty is not 
lowered, so the attacker ends up mining less
official blocks per unit of time in Strategy 1 compared to Strategy 2.

\medskip

On the contrary, if the attacker decides to avoid competitions with honest
miners as much as possible, he will earn less inclusion rewards (he can even
decide to ignore totally these rewards) but his speed of validation of blocks will
increase. So, what is the best
strategy will depend very sensitively on the parameters of the reward system.

\medskip

As explained in {\cite{GPM18e}}, the correct benchmark to compare profitabilities of
two strategies is the revenue ratio 
$$
\Gamma = \frac{\mathbb{E}
[R]}{\mathbb{E} [T]}
$$ 
where $R$ is the revenue of the miner per attack cycle
and $T$ is the duration of an attack cycle. In Bitcoin, after a difficulty
adjustment, this quantity becomes in the long run proportional to 
$$
\tilde \Gamma =\frac{\mathbb{E}
[R_s]}{\mathbb{E} [L]}
$$ 
where $L$ (resp. $R_s$) is the number of new blocks
(resp. new blocks mined by the attacker) added to the official blockchain per
attack cycle. The difficulty adjustment is not continuous in Bitcoin as it is updated
every 2016 official new blocks. With the martingale tools introduced in {\cite{GPM18}},
we computed how long it takes for the attack to become profitable (this computation is not possible 
with the old Markov chain model).

\medskip

In Ethereum, the difficulty adjustment formula is different. The revenue 
ratio is proportional to 
$$
\tilde \Gamma=\frac{\mathbb{E} [R]}{\mathbb{E} [L]+\mathbb{E} [U]}
$$ 
where $U$ is the number of referred uncles 
and $R$ is the total revenue of the attacker in the attack cycle. Moreover, 
the revenue $R$ per attack cycle has three different contributions : 
\begin{enumerate}
 \item The revenue $R_s$ coming from ``static'' blocks.
 \item The revenue $R_u$ coming from ``uncles'' blocks.
 \item The revenue $R_n$ coming from ``nephews'' blocks.
\end{enumerate}
In Bitcoin's revenue analysis only $R_s$ is present.
Therefore, for Ethereum we have
$$
\tilde \Gamma=\frac {\mathbb{E} [R]}{\mathbb{E} [L]+\mathbb{E} [U]}=
\frac{\mathbb{E} [R_s] +\mathbb{E} [R_u] +\mathbb{E} [R_n]}{\mathbb{E} [L]+\mathbb{E} [U]}
$$ 
The new terms on the numerator $\mathbb{E} [R_u]$ and $\mathbb{E} [R_n]$ 
increase the revenue of  the attacker and are incentives for block withholding 
attacks. On the other hand, the new term $E[U]$ in the denominator plays against the 
profitability of the attack and tends to mitigate the attack. Only 
an exact computation of these terms can show which one is the most profitable 
strategy. Another particularity of Ethereum is the continuous adjustment of the difficulty. 
Thus a block-witholding attack is very quickly profitable.

\medskip

There are other selfish mining strategies in Ethereum. For
instance, the attacker can publish his secret blocks slowly, two
by two, instead of one by one. In this article we limit our study to
Strategy 1 and Strategy 2. The main result are the closed-form formulas 
for the apparent hashrates in Strategy 1 and 2. The main conclusion is 
that the effect on the difficulty adjustment is prevalent, so that Strategy 2 outperforms Strategy 1.

\section{A combinatorics approach}

In this section we present a general setup that is common for all strategies. 
We apply our combinatorics approach to selfish
mining as done previously for Bitcoin {\cite{GPM19a}}. Dyck words and Catalan numbers are 
a powerful tool to compute the
revenue ratio of a selfish miner in Bitcoin. In {\cite{GPM19a}} we proved the following Theorem and Corollary:

\medskip

\begin{theorem}\label{rapl}
  Let $L$ be the number of official new blocks added to the
  official blockchain after an attack cycle. We have 
  \begin{align*}
  \mathbb{P} [L = 1] &= p \ , \\
  \mathbb{P} [L = 2] &= pq + pq^2 \ , \\ 
  \end{align*}
  and for $n \geq 3$, 
  $$
  \mathbb{P} [L =n] = pq^2  (pq)^{n - 2} C_{n - 2}
  $$ 
  where $C_n = \frac{(2 n) !}{n! (n + 1)!}$ is the n-th Catalan number.
\end{theorem}
\medskip

\begin{corollary}
  We have $\mathbb{E} [L] = 1 + \frac{p^2 q}{p - q}$.
\end{corollary}

\medskip

We can represent the combinatorics information of an attack cycle $\omega$ by the chronological sequence 
of blocks, S (for Selfish) and H (for Honest).
The relation between selfish mining and Dyck words is the following (see \cite{GPM19a}),

\medskip

\begin{proposition}\label{bcdyck}
  Let $\omega$ be an attack cycle starting with SS. Then, $\omega$ ends with H and 
  the intermediate sequence $w$ defined by $\omega = SSwH$ is a Dyck word.
\end{proposition}

\medskip

\begin{definition}
  For  $n \geq 0$, we
  denote by $C_n (x) = \sum_{k = 0}^n C_k x^k$, the $n$-th partial sum of the
  generating series of the Catalan number.
\end{definition}

\medskip

\begin{example}\label{Catalan4}\normalfont
  We have $C_4 (x) = 1 + x + 2 x^2
  + 5 x^3 + 14 x^4$.
\end{example}

\medskip

\begin{definition}
  We define $\pi_0 = \pi_1 = 0$ and for $k \geq 2$, 
  $$
  \pi_k = pq^2  ({\bf{1}}_{k = 2} +{\bf{1}}_{k \geq 2} \cdot (pq)^{k - 2} C_{k - 2} ) \ .
  $$
\end{definition}

\medskip

The following lemma results from Theorem \ref{rapl}.

\medskip

\begin{lemma}\label{pisecpi2}
  Let $\omega$ be an attack cycle.
  \begin{itemize}
   \item For $k \geq 0$, the probability that 
  $\omega$ is won by the attacker and  
   $L (\omega) = k$ is $\pi_k$.
   \item For $k \geq 2$, the probability that 
   $\omega$ is won by the attacker and
   $L (\omega) \leq k$ is $pq^2 + pq^2 C_{k - 2} (pq)$.
  \end{itemize}

\end{lemma}

\begin{proof}
  We have either $\omega = \text{SHS}$ or $\omega$ starts with SS. The result
  then follows from Lemma \ref{pbdn} in the Appendix.
\end{proof}

\medskip

For Ethereum, the ``static'' part $R_s$ of the revenue of the selfish
miner coming from rewards for validated blocks  is the same as for Bitcoin. However, we need to add the new terms
$R_s$ and $R_n$ coming from uncle and nephew rewards.

\medskip

\begin{definition}\label{gedef}
  If $\omega$ is an attack cycle, we denote by $U (\omega)$  (resp. $U_s
  (\omega)$, $U_h (\omega)$) the random variable counting the number of 
  uncles created during the cycle $\omega$ which
  are referred by nephew blocks (resp. nephew blocks mined by the selfish
  miner, nephew blocks mined by the honest miners) in the cycle $\omega$ or in a later
  attack cycle.
  
  We denote by $V (\omega)$ the random variable counting the number of
  uncles created during the cycle $\omega$ and are referred by nephew blocks
  (honest or not) in an attack cycle strictly after $\omega$.
\end{definition}

\medskip

We take from \cite{NF19} the notation $K_u$ for the uncles reward function,  
and we denote by $\pi$ the inclusion reward (see the glossary at the end).

For a general block withholding strategy, the random variables from Definition \ref{gedef} do not contain all the information 
for the computation of the attacker's revenue. It depends not only on the number of uncles mined by the attacker 
but also on their distance $d$ to its corresponding nephews. 

However, for a miner following a selfish mining strategy, the part of his revenue coming from uncle 
rewards are easy to compute, as shown in the next Proposition, because only the case $d=1$ is possible. 
This observation was already made in \cite{NF19}.
  
\medskip

\begin{proposition}\label{ERU}
  Let $R_u(\omega)$ be the total amount of uncle rewards of the selfish miner during an attack cycle $\omega$. 
  We have: 
  $$
  \mathbb{E}[R_u] = p^2 q (1 - \gamma) K_u (1) \ .
   $$ 
\end{proposition}

\medskip

Currently on Ethereum we have $K_u (1) = \frac{7}{_{} 8} b$.

\medskip

\begin{proof}
  Let $\omega$ be an attack cycle. If $\omega = \text{SHH}$ with a second
  honest block mined on top of another honest block after a competition, the attacker has an uncle
  which is referred by the second honest block of the honest miners in the cycle
  $\omega$. Otherwise, if $\omega \not= \text{SHH}$ then the attacker has no
  uncle in the cycle $\omega$ (the only uncle blocks are those mined by the honest miners).
\end{proof}

\medskip

The \textit{apparent hashrate} is the long term apparent hashrate of the attacker after the manipulation of the difficulty by the attacker.

\medskip

\begin{definition}
We denote by $\widetilde{q_{}}_B$, resp. $\widetilde{q_{}}_E$, the long term apparent hashrate of the selfish miner in Bitcoin, resp. Ethereum, defined by
\begin{align*}
 \tilde{q}_B &=\frac{\mathbb{E} [R_s]}{\mathbb{E} [L]} \\
 \widetilde{q}_E & =\frac{\mathbb{E} [R_s] +\mathbb{E} [R_u] +\mathbb{E} [R_n]}{\mathbb{E} [L] +\mathbb{E} [U]}\\
\end{align*}
\end{definition}

\medskip

For Bitcoin we have the following formula (see \cite{ES14} and \cite{GPM18}),

\begin{equation*}
\widetilde{q_{}}_B = \frac{[(p - q) (1 + pq) + pq] q - (p - q) p^2 q (1 - \gamma)}{pq^2 + p - q}
\end{equation*}

For Ethereum only $\mathbb{E} [U]$ and $\mathbb{E} [U_s]$ are relevant for the computation of the apparent hashrate of the selfish miner:

\medskip

\begin{theorem}\label{qetilde}
  We have
  \begin{equation*}
    \widetilde{q_{}}_E  =  \tilde{q}_B \cdot \frac{\mathbb{E}
    [L]}{\mathbb{E} [L] +\mathbb{E} [U]} + \frac{p^2 q (1 - \gamma) K_u
    (1)}{\mathbb{E} [L] +\mathbb{E} [U]}+  \frac{\mathbb{E} [U_s]}{\mathbb{E} [L] +\mathbb{E} [U]} \, \pi \ .
  \end{equation*}
\end{theorem}
Currently on Ethereum we have $K_u (1) = \frac{7}{8}$ and $\pi = \frac{1}{32}$.
\begin{proof}
  Using Proposition \ref{ERU}, we have:  
  \begin{align*}
    \widetilde{q_{}}_E & =\frac{\mathbb{E} [R_s] +\mathbb{E} [R_u] +\mathbb{E} [R_n]}{\mathbb{E} [L] +\mathbb{E} [U]}\\
    & =\frac{\mathbb{E} [R_s]}{\mathbb{E} [L]} \cdot \frac{\mathbb{E}
    [L]}{\mathbb{E} [L] +\mathbb{E} [U]} + \frac{\mathbb{E}
    [R_u]}{\mathbb{E} [L] +\mathbb{E} [U]} +  \frac{\mathbb{E} [U_s]}{\mathbb{E} [L] +\mathbb{E} [U]} \, \pi \\
    & = \tilde{q}_B \cdot \frac{\mathbb{E} [L]}{\mathbb{E} [L] +\mathbb{E} [U]} + 
    \frac{p^2 q (1 - \gamma) K_u (1)}{\mathbb{E} [L] +\mathbb{E} [U]} +  \frac{\mathbb{E} [U_s]}{\mathbb{E} [L] +\mathbb{E} [U]} \, \pi
  \end{align*}
\end{proof}

In next sections we compute $\mathbb{E} [U_s]$ and $\mathbb{E} [U]$ for different selfish mining strategies.

\section{Strategy 1: Maximum belligerence signalling all uncles.}

We consider here the strategy described in \cite{NF19} where the attacker engages in competition with the honest miners as
often as possible, and signals all possible ``uncles'.

\subsection{General definitions and basic results.}

\medskip

\begin{definition}\label{defh}
  The relative height of an orphan block $\mathfrak{b}$ validated
  by the honest miners is the difference between the height of the secret fork
  of the attacker at the time of creation of $\mathfrak{b}$ and the height of
  $\mathfrak{b}$. We denote it $h (\mathfrak{b})$.
\end{definition}

\medskip

\begin{example}\normalfont
  For $\omega = \text{SSSHSHSHH}$, the first three ``honest'' blocks have
  relative height equal to 2 and the last ``honest'' block has a relative
  height equal to 1.
\end{example}

\medskip

\begin{proposition}\label{hbd}
  Let $\mathfrak{b}$ be an uncle block mined by an honest miner and
  signaled by a nephew block which is at a distance $d$ of $\mathfrak{b}$. Then,
  we have \ $h (\mathfrak{b}) < d$.
\end{proposition}

\medskip

\begin{proof}
  Let $\mathfrak{b}'$ be the last block mined by the selfish miner at the date
  of creation of $\mathfrak{b}$. Notice that $h (\mathfrak{b})$ is also the
  number of blocks between $\mathfrak{b}$'s parent and $\mathfrak{b}'$. Thus the
  distance between $\mathfrak{b}$ and a possible nephew is necessarily strictly
  greater than $h (\mathfrak{b})$.
\end{proof}

\medskip

\begin{note}\normalfont
  \label{omw}
  Let $n \geq 0$ and $\omega = \text{SS} w$ be an attack
  cycle with $w = w_1 \ldots w_{2 n + 1}$, $w_i \in \{ S, H \}$ and $w_{2 n +
  1} = H$. Then, $w$ can be identified with a simple finite path $(X_i)_{0
  \leqslant i \leqslant 2 n + 1}$ starting from $0$, satisfying: $\forall i
  \leqslant 2 n + 1, X_i = X_{i - 1} + 1$ (resp. $X_i = X_{i - 1} - 1$) if
  $w_i = S$ (resp. $w_i = H$) and ending at $X_{2 n + 1} = - 1$ (see the
  Appendix). The index $i$ indicates the $(i + 2)$-th block validated during
  $\omega$. It has been mined by the attacker (resp. honest miners) if $X_i =
  X_{i - 1} + 1$ (resp. $X_i = X_{i - 1} - 1$).
\end{note}

\medskip

\begin{proposition}
  \label{omw2}Let $\omega = \text{SS} w$ an attack cycle starting with two S
  with $w = w_1 \ldots w_{2 n + 1}$, $w_i \in \{ S, H \}$ and $w_{2 n + 1} =
  H$. \ We denote by $X : \left[ 0, 2 n + 1 \right] \longrightarrow
  \left[ - 1, + \infty \right]$ the path associated with $w$ as in
  Note \ref{omw}. For $i \leqslant 2 n + 1$, let $\mathfrak{b}_i$ denote the
  $i$-th validated block in $w$. Then we have:
  \[ X_i < X_{i - 1} \Longrightarrow h (\mathfrak{b}_i) = X_i + 2 \]
\end{proposition}

\medskip

\begin{proof}
  By induction on $i$, we show that $X_i + 2$ represents the advance of the
  fork created by the attacker over the official blockchain at the time of
  creation of the $i$-th block in $w$. Now, if $X_i < X_{i - 1}$ then by Note
  \ref{omw}, $\mathfrak{b}_i$ is a block validated by the honest miners. So $h
  (\mathfrak{b}_i)$ is well defined, and we get the result using Definition
  \ref{defh}.
\end{proof}

\medskip

\begin{proposition}\label{no}
  Let $\omega = \text{SSw}$ be an attack cycle starting with two S
  and let $\mathfrak{b}_i$ be the $i$-th block validated in $w$. We denote by
  $X$ the associated path according to Note \ref{omw}. If $\mathfrak{b}_i$ is
  an uncle then we have:
  \begin{enumerate}
    \item \label{condition_1} $X_i < n_1 - 2$
    \item \label{condition_2} $X_i < X_{i - 1}$
  \end{enumerate}
\end{proposition}

\medskip

\begin{proof}
  This follows from Proposition \ref{hbd} and Proposition \ref{omw2}.
\end{proof}

\medskip

\begin{definition}
  If $\omega = \text{SS} w$ is an attack cycle starting with two blocks S,
  then we denote by $H (\omega)$ the random variable counting the number of blocks in the cycle $\omega$ 
  fulfilling (\ref{condition_1}) and (\ref{condition_2}) from Proposition \ref{no}.
\end{definition}

\medskip

If $w$ is an attack cycle, the condition $\omega = \text{SS} \ldots$ means
that $\omega$ starts with two $S$.

\medskip

\begin{proposition}
  \label{eho}We have:
  \[ \mathbb{E} [H(\omega) | \omega = \text{SS} \ldots] = \frac{p}{p - q} 
     \left( 1 - \left( \frac{q}{p}^{} \right)^{n_1 - 1} \right) \]
\end{proposition}

\medskip

\begin{proof}
  See Lemma \ref{maon} in the Appendix.
\end{proof}

\medskip

\subsection{Expected number of referred uncles by attack cycle}

We can be more precise in Proposition \ref{no}.

\medskip

\begin{lemma}\label{poun}
  Let $\omega = \text{SSw}$ and $X$ be the associated path
  from Note \ref{omw}. We denote by $\mathfrak{b}_i$ the $i$-th block
  in $w$ and suppose that conditions (\ref{condition_1}) and (\ref{condition_2}) 
  from Proposition \ref{no} are satisfied. The
  probability for $\mathfrak{b}_i$ to be an uncle is equal to $\gamma$, except
  when $\mathfrak{b}_i$ is the first block validated by the honest
  miners, then this probability is $1$.
\end{lemma}

\medskip

\begin{example}\normalfont
Suppose that $n_1 = 4$ and let $\omega = \text{SS}w$ with $w =\text{SHSSSHHHH}$. 
The blocks validated by the honest miners correspond to an index 
$i\in E=\{2, 6, 7, 8, 9\}$. 
We have $X_6=2$ and $X_i<2$ for $i\in E$ and $i\not=6$.
The first block validated by the honest miners is an uncle with probability $1$. 
The second block validated by the honest miners is a stale block which cannot be referred by a nephew block. 
All other blocks validated by the honest miners in $\omega$ can be uncles with probability $\gamma$. 
Note also that the last three blocks of the honest miners are not referred in $\omega$ and
will be referred by the first future official block of the next attack cycle.

\end{example}

\medskip

Using these observations, we can now compute $\mathbb{E} [U]$.

\medskip

\begin{proposition}\label{proeu}
  We have:
  \begin{equation*}
   \mathbb{E} [U] = q + \frac{q^3 \gamma}{p - q} - \frac{p^3}{p - q} 
     \left( \frac{q}{p}^{} \right)^{n_1 + 1} \gamma - q^{n_1 + 1}  (1 -
     \gamma) 
  \end{equation*}
\end{proposition}

\medskip

\begin{proof}
  If $\omega = H$, then $U (\omega) = 0$. If $\omega \in \{ \text{SHS},
  \text{SHH} \}$, then, $U = 1$. Otherwise, $\omega$ starts with two consecutive S. 
  Then, by Proposition \ref{eho} and Lemma \ref{poun}, we have,
  \begin{equation*}
    \mathbb{E} [U]  =  (0 \cdot p) + 1 \cdot (pq^2 + p^2 q) + (\mathbb{E}
    [H (\omega) | \omega = \text{SS} \ldots] \gamma + (1 - \gamma)  (p + pq +
    \ldots + pq^{n_1 - 2})) \cdot q^2 
  \end{equation*}
  The last term comes from the following fact: 
  When the first honest block present in $\omega$ corresponds to an
  index $i$ satisfying $X_i < n_1 + 2$, then its contribution to $\mathbb{E}
  [U]$ is underestimated by $\mathbb{E} [H (\omega) | \omega = \text{SS}
  \ldots] \gamma$ because it has probability $1$ to be an uncle. This only occurs when
  $\omega$ starts with SS...SH with the first $k$ blocks validated by the
  selfish miner with $k \leqslant n_1$, from where we get the last term. In conclusion we have:
  \begin{equation*}
    \mathbb{E} [U]  =  pq + \left( \frac{p}{p - q}  \left( 1 - \left(
    \frac{q}{p}^{} \right)^{n_1 - 1} \right) \gamma \right) \cdot q^2 \\
    +  (1 - \gamma)  (1-q^{n_1 - 1}) \cdot q^2
  \end{equation*}
  and we get the result by rearranging this last equation.
\end{proof}

\medskip

\begin{note}\normalfont
  In particular, we obtain $\underset{n_1 \rightarrow \infty}{\lim} \mathbb{E}
  [U] = q + \frac{q^3 \gamma}{p - q}$. This limit can also be derived by
  observing that if $n_1 = \infty$, then $\mathbb{E} [U|L = n] = 1 + \gamma
  (n - 2)$ and using Theorem \ref{rapl}.
\end{note}

\medskip

Now, we compute the expected number of uncles per
attack cycle which are referred by nephews (honest or not) belonging to the
next attack cycle.

\medskip

\begin{lemma}\label{pssk}
  The probability for an attack cycle to end with exactly $k$ consecutive appearances of ``H''
  with $k \geq 1$, conditional that it starts with SS, is $pq^{k - 1}$.
\end{lemma}

\medskip

\begin{proof}
  Let $k \geq 1$. An attack cycle $\omega$ ends with exactly
  $k$ consecutive appearances of ``H'' if and only if 
  $\omega = \text{SSwH}$ where $w$ is a Dyck word that
  ends with exactly $k - 1$ ``H''. The result then follows from Appendix, Lemma
  \ref{pdyck}.
\end{proof}

\medskip

\begin{proposition}
  \label{ev}We have:
  \[ \mathbb{E} [V] = \frac{q^2}{p}  (1 - q^{n_1 - 1}) \gamma + (1 - \gamma)
     pq^2  \, \frac{1 - (pq)^{n_1 - 1}}{1 - pq} \]
\end{proposition}

\medskip

\begin{proof}
  If an attack cycle $\omega$ does not start with two S, then $V (\omega) =
  0$. If $\omega$ starts with two ``S'' and ends with
  exactly $k$ ``H'' in a row ($k \geq 1$), then only the last $n_1 - 1$
  blocks can be uncles signaled by future blocks. This happens with
  probability $\gamma$ for each block H in this sequence, except for
  the first block validated by the honest miners if it belongs to this
  sequence. In this last case, $\omega = \text{SS} \ldots \text{SH}
  \ldots \text{HH}$ with at most $n_1$ letters S and $n_1 - 1$ letters ``H''.
  So, by Lemma \ref{pssk}, we have
  \begin{equation*}
    \mathbb{E} [V]  =  q^2  \sum_{k \geq 1} \inf (k, n_1 - 1)
    pq^{k - 1} \gamma +  (1 - \gamma) q \sum_{k = 1}^{n_1 - 1} (pq)^k
  \end{equation*}
\end{proof}

\medskip

\subsection{Expected revenue of the selfish miner from inclusion rewards.}

We compute now $\mathbb{E} [U_h]$.

\medskip

\begin{proposition}\label{euh}
  We have:
  \begin{equation*}
    \mathbb{E} [U_h] =  p^2 q + (p + (1 - \gamma) p^2 q) 
    \left( \frac{q^2}{p}  (1 - q^{n_1 - 1}) \gamma + (1 - \gamma)
    pq^2  \, \frac{1 - (pq)^{n_1 - 1}}{1 - pq} \right)
  \end{equation*}
\end{proposition}

\medskip

\begin{proof}
  We have $U_h (\omega) = U_h^{(1)} (\omega) + U_h^{(2)} (\omega)$ where
  $U_h^{(1)} (\omega)$ (resp. $U_h^{(2)} (\omega)$) counts the number of
  uncles referred by honest nephews only present in $\omega$ (resp. in the next
  attack cycle after $\omega$). It is clear that $U_h^{(1)} (\text{SHH}) = 1$
  and $U_h^{(1)} (\omega) = 0$ if $\omega \not= \text{SHH}$. So,
  \begin{equation}
    \mathbb{E} [U_h^{(1)}] = p^2 q \label{euh1}
  \end{equation}
  Moreover, given $\omega$, the probability that H is the next official block
  after $\omega$ is $p + (1 - \gamma) p^2 q$. This happens if and only if the
  next attack cycle is either H or SHH. If this event occurs, then the first honest block
  in the next attack cycle will signal the previous uncles created in $\omega$.
  Therefore, we have
  \begin{equation}
    \mathbb{E} [U_h^{(2)}] = (p + (1 - \gamma) p^2 q) \cdot \mathbb{E} [V]
    \label{euh2e}
  \end{equation}
  Hence, we get the result by (\ref{euh1}), (\ref{euh2e}) and Proposition
  \ref{ev}.
\end{proof}

\medskip

\begin{corollary}
  \label{coeua}We have
  \begin{align*}
    \mathbb{E} [U_s]  =& q + \frac{q^3 \gamma}{p - q} - \frac{pq^2}{p - q} 
    \left( \frac{q}{p}^{} \right)^{n_1 - 1} \gamma - q^{n_1 + 1}  (1 -
    \gamma) \\ &-  \left[ p^2 q + (p + (1 - \gamma) p^2 q)  \left( \frac{q^2}{p}  (1 -
    q^{n_1 - 1}) \gamma + (1 - \gamma) pq^2  \frac{1 - (pq)^{n_1 - 1}}{1 - pq}
    \right) \right]
  \end{align*}
\end{corollary}

\medskip

\begin{proof}
  With the same notations as above, we have: $U (\omega) = U_s
  (\omega) + U_h (\omega)$ and we use Proposition \ref{proeu} and Proposition
  \ref{euh}.
\end{proof}

\medskip

\subsection{Apparent hashrate of Strategy 1}

Using Theorem \ref{qetilde}, Proposition \ref{proeu} and Corollary
\ref{coeua} we can plot the region of $(q, \gamma) \in [0, 0.5] \times [0, 1]$
of dominance of the selfish mining Strategy 1 (SM1) over the honest strategy. This corresponds to
$\widetilde{q_{}}_E > q$. We obtain Figure 1.

\begin{figure}[h]
  \resizebox{180pt}{180pt}{\includegraphics{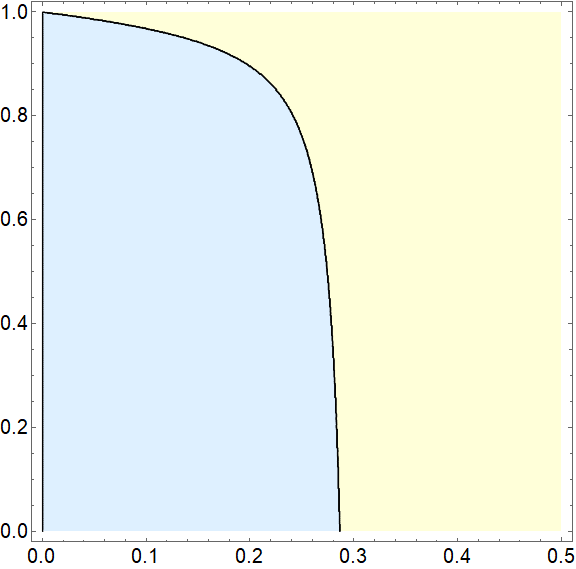}}
  \put(-135, 85){HM}
  \put(-55, 85){SM1}
  \caption{Comparing HM and SM1 strategies.}
\end{figure}

\medskip

We compute now the expected revenue of the honest miners by attack cycle. 
We compute first the expected distance between uncles and nephews by attack cycle.

\subsection{Expected distance between uncles and nephews by attack cycle}

\medskip

  If $\mathfrak{b}$ is an uncle, we denote by $\delta(\mathfrak{b})$ the distance 
  between $\mathfrak{b}$ and its nephew.
  We start by a remark.

\begin{remark}\label{remm}
  Let $\mathfrak{b}$ be an orphan block validated by the honest miners as in Definition \ref{defh}. If $\mathfrak{b}$ is an uncle then 
  $\delta (\mathfrak{b}) = h(\mathfrak{b}) + 1$.
\end{remark}

\begin{definition}
  If $\omega = \text{SS} w$ is an attack cycle starting with two blocks S,
  we set
  $$ D (\omega) = \sum _{\mathfrak{b}} \left( h(\mathfrak{b}) + 1\right) 
  $$
  where the sum is taken over all honest blocks $\mathfrak{b}$ in $\omega$ fulfilling
  Conditions  (\ref{condition_1}) and (\ref{condition_2}) from Proposition \ref{no}.
\end{definition}

\begin{proposition}\label{edo}
  We have:
  \[ \mathbb{E} [D(\omega) | \omega = \text{SS} \ldots] 
   = \frac{p}{(p-q)^{2}} \left(2 p - q - 
   \bigl( p + n_1 (p-q)\bigr)\cdot \left( \frac{q}{p}\right)^{n_1 - 1}
   \right)
   \]
\end{proposition}

\begin{proof}
 Let $\omega = \text{SS} w$ be an attack cycle starting with two S 
 with $w=w_1\ldots w_{\nu}$
 and let $X$ be the associated path according to Note \ref{omw}. In particular, we have 
 $X_{\nu}=-1$ and $X_i\geq 0$ for $i<\nu$. By Proposition \ref{omw2} and
 Lemma \ref{maon} in the Appendix, we have:
 \begin{align*}
  \mathbb{E} [D(\omega) | \omega = \text{SS} \ldots] &=
  \mathbb{E} \left[\sum_{i=1}^{\nu}  \left(X_i +3\right)\cdot {\bf 1}_{(X_i < n_1 -2) 
  \wedge (X_i < X_{i - 1})} \right]\\
  &= \mathbb{E} \left[\sum_{i=1}^{\nu}  X_i \cdot {\bf 1}_{(X_i < n_1 -2) 
  \wedge (X_i < X_{i - 1})} \right] + 3 
  \mathbb{E} \left[\sum_{i=1}^{\nu}  \cdot {\bf 1}_{(X_i < n_1 -2) 
  \wedge (X_i < X_{i - 1})} \right]\\
  &= 
  \frac{p}{(p-q)^2} 
  \left( 2q-p - 
  \bigl(q+ (n_1 -2) (p-q)\bigr)\cdot \left( \frac{q}{p}\right)^{n_1 - 1} \right)
  + \frac{3 p}{p - q}  \left( 1 -
  \left( \frac{q}{p}^{} \right)^{n_1 - 1} \right)
  \\
  &= \frac{p}{(p-q)^2} \left(
  2 q -p + 3 (p-q) - \left( q-2 (p-q)+n_1\bigl(p-q\bigr) 
  +3  \bigl(p-q\bigr) \right)
  \cdot \left( \frac{q}{p}\right)^{n_1 - 1}
  \right)
 \end{align*}
  Hence, we get the result.
 \end{proof}

\begin{definition}\label{defDelta}
  Let $\omega$ be an attack cycle. We set
  $$
  \Delta (\omega) =  \sum _{\mathfrak{b}} \delta (\mathfrak{b}) 
  $$
  The last sum being taken over all refered uncles in $\omega$.
\end{definition}

\begin{proposition}\label{pound}
  We have 
  \begin{align*}
  \mathbb{E} [\Delta] = &  p q+\frac{p q^{2} \gamma}{(p-q)^{2}}
  \left(2 p - q - 
   \bigl( p + n_1 (p-q)\bigr)\cdot \bigl( \frac{q}{p}\bigr)^{n_1 - 1}
   \right)\\
   &+\frac{(1-\gamma) q}{p}\left(q (1+p)-(1 + n_1 p) q^{n_1}
   \right)
  \end{align*}
\end{proposition}

\begin{proof}
  We proceed as in the proof of Proposition \ref{proeu}.
  If $\omega = H$, then $\Delta (\omega) = 0$. If $\omega \in \{ \text{SHS},
  \text{SHH} \}$, then, $\Delta (\omega) = 1$. Otherwise, $\omega$ starts with two consecutive S. 
  Then, using Lemma \ref{poun}, we get
  \begin{equation*}
    \mathbb{E} [\Delta]  =  pq^2 + p^2 q + (\mathbb{E}
    [D (\omega) | \omega = \text{SS} \ldots] \gamma + (1 - \gamma)  (2p + 3pq +
    \ldots + n_1 pq^{n_1 - 2})) \cdot q^2 
  \end{equation*}
  The last term comes from the following fact: 
  when the first honest block present in $\omega$ corresponds to an
  index $i$ satisfying $X_i < n_1 + 2$, then its contribution to $\mathbb{E}
  [\Delta]$ is underestimated by $\mathbb{E} [D (\omega) | \omega = \text{SS}
  \ldots] \gamma$ because it has probability $1$ to be an uncle. This only occurs when
  $\omega$ starts with SS...SH with the first $k$ blocks validated by the
  selfish miner with $k \leqslant n_1$, from where we get the last term. We have:
    \begin{align*}
    2q + 3 q^{2} + \ldots + n_1 q^{n_1 - 1} &= -1+\left(
    \frac{q^{n_1 +1}-1}{q-1}\right)' = -1+\left(
    \frac{q^{n_1 +1}}{q-1}\right)' - \left(\frac{1}{q-1}\right)'\\
     &= -1 + \frac{(n_1 +1)q^{n_1}}{q-1} - \frac{q^{n_1 +1}}{(q-1)^{2}}
     + \frac{1}{(q-1)^{2}}\\
     &= -1+\frac{1}{p^{2}}+\frac{q^{n_1}}{(q-1)^{2}}
     \left( (n_1 +1) (q-1) -q\right)\\
     &= \frac{1-p^{2}}{p^{2}}-\frac{q^{n_1}}{p^{2}}
     \left( q+(n_1 +1)p\right)\\
     &= \frac{q (1+p)}{p^{2}}-\frac{q^{n_1}}{p^{2}}
     \left( 1+n_1 p\right)
  \end{align*}
  So,
  \begin{equation}\label{Del}
  \left(2p + 3 p q + \ldots + n_1 p q^{n_1 - 2}\right) q^{2}
  = \frac{q}{p} \left(q (1+p) - \bigl(1+n_1 p\bigr) q^{n_1}\right)
  \end{equation}
  Hence we get the result using Proposition \ref{edo}.
\end{proof}

\medskip

\subsection{Deflation}
  With the new difficulty adjustment formula, the duration time of an attack cycle in
  Ethereum is $(\mathbb{E} [L] + \mathbb{E} [U]) \tau_1$ where $\tau_1$ is the mean
  interblock time in Ethereum (which is currently $15$ seconds). The number of coins created 
  in an attack cycle is 
  $\left(\mathbb{E} [L] + \frac{7}{8}\mathbb{E} [U] - \frac{1}{8} \mathbb{E} [\Delta] + 
  \mathbb{E} [U] \pi\right) b$ where $b$ is the coinbase in Ethereum.
  Thus, on average, there is a monetary creation of 
  $$
  \displaystyle\frac{\mathbb{E} [L] +
  \bigl(\frac{7}{8}+\pi\bigr)\mathbb{E} [U]- 
  \frac{\mathbb{E} [\Delta]}{8}}{\mathbb{E} [L] +
   \mathbb{E} [U]} \, b
   $$ 
   for every inter-block time $\tau_1$, 
   whereas without selfish miner, it  
   is only $b$ on average. So, selfish mining leads to a deflation index
   \begin{equation}\label{deflation}
   \iota = \displaystyle\frac{
  \bigl(\frac{1}{8}-\pi\bigr)\mathbb{E} [U]+ 
  \frac{\mathbb{E} [\Delta]}{8}}{\mathbb{E} [L] +
   \mathbb{E} [U]}
   \end{equation}
   Currently we have $\pi=\frac{1}{32}$, thus $\iota>0$.

\medskip

\subsection{Apparent hashrate of the honest miners}
  Let $\tilde{p}$ be the apparent hashrate of the honest miners 
  in presence of a selfish miner. We have 
   \begin{equation}\label{pqtilde}
   \tilde{p} + \tilde{q} = 1 - \iota
   \end{equation}
   where $\tilde{q}$ is the apparent hashrate of the selfish miner. 
   We observe numerically that $\tilde{q} > q - \iota$ 
   for any values of $(q,\gamma)$. 
   So, even if the attack is not profitable for the selfish miner (case $\tilde{q} < q$)
   we have $\tilde{p} < p$ which means that the honest miners are impacted by 
   the presence of a selfish miner in the network.

\medskip

\section{Strategy 2A: Brutal Fork signaling all uncles.}

\medskip

We study now another Selfish Mining Strategy (Strategy 2 or SM2): Brutal fork. 
In this case, the attacker keeps
secret his blocks as long as possible and only releases its fork, all at once, at the end of the attack cycle. We call this strategy "brutal fork" because this leads, periodically, to deep reorganizations of the official blockchain.
Strategy 2A (or SM2A) corresponds to the case when also the attacker refers all possible uncles.

\medskip

\begin{proposition}\label{proeu2}
  We have $\mathbb{E} [U] = q - q^{n_1 + 1}$.
\end{proposition}

\medskip

\begin{proof}
  We have $U = 0$ if and only if the attack cycle is H or if it starts with $n_1 + 1$ blocks of type S.
  Otherwise, we have $U = 1$. So,
  \begin{equation*}
    \mathbb{E} [U]  =  \mathbb{P} [U > 0] =  1 - (p + q^{n_1 + 1}) =  q - q^{n_1 + 1}
  \end{equation*}
\end{proof}

\medskip

We compute now $\mathbb{E} [V]$

\medskip

\begin{proposition}\label{peus}
  We have $\mathbb{E}[V] = p q^{2}\cdot \frac{1-(p q)^{n_1 - 1}}{1-p q}$.
\end{proposition}

\medskip

\begin{proof}
We have $V=1$ if and only if the attack cycle $\omega$ is SS..SH..H with $2\leq k\leq n_1$ S. In that case, the first H is an uncle signaled by the first future official block in the attack cycle after $\omega$. Otherwise, $V=0$. 
So, $\mathbb{E}[V] = pq^{2}+\ldots+p^{n_1-1}q^{n_1}$, and we get the result.
\end{proof}

\medskip

\begin{proposition}\label{euh2}
  We have $\mathbb{E}[U_h] = p^{2}q+\left( p + (1-\gamma)p^{2}q\right) p q^{2}\cdot \frac{1-(p q)^{n_1 - 1}}{1-p q}$.
\end{proposition}

\medskip

\begin{proof}
  The proof is almost identical as the proof of Proposition \ref{euh}. If $U_h^{(1)} (\omega)$ (resp. $U_h^{(2)} (\omega)$) 
  counts for the number of uncles referred by honest nephews only present in $\omega$ (resp. in the attack cycle just after $\omega$),
  then we have $\mathbb{E} [U_h^{(1)}] = p^2 q$, $\mathbb{E} [U_h^{(2)}] = \left( p + (1 - \gamma) p^2 q \right) \cdot \mathbb{E} [V]$
  and $U_h = U_h^{(1)} + U_h^{(2)}$. The only difference is the value of $\mathbb{E} [V]$ which this time is given  by
  Proposition \ref{peus}, and we get the result.
\end{proof}

\medskip

\begin{corollary}\label{coeub}
  We have
  \begin{equation*}
    \mathbb{E} [U_s]  =  \mathbb{E} [U]  - \mathbb{E} [U_h] =  q - q^{n_1 + 1} - \left( p^{2}q+
    \left( p + (1-\gamma)p^{2}q\right) p q^{2}\cdot \frac{1-(p q)^{n_1 - 1}}{1-p q}\right)
  \end{equation*}
\end{corollary}

\medskip

\begin{note} \normalfont
  When $\gamma = 0$, the two strategies 1 and 2A are identical:
  in both cases, the honest miners always build blocks on top of honest blocks. 
  So, $\mathbb{E} [U]$, $\mathbb{E} [U_h]$ and $\mathbb{E} [U_s]$ must coincide for $\gamma = 0$. 
  We can check in the different formulas that this is the case. 
  See Propositions \ref{proeu}, \ref{euh},\ref{proeu2}, \ref{euh2} 
  and Corollaries \ref{coeua}, \ref{coeub}.
\end{note}

\medskip

\subsection{Apparent hashrate of Strategy 2A}

We use again Theorem \ref{qetilde} and we plot in parameter space 
in Figure 2 the region of $(q, \gamma)\in [0, 0.5] \times [0, 1]$ comparing  Selfish Mining Strategy 2A to the honest strategy. 

\medskip

We observe that if $\gamma=0$ then we have SM2A is superior to honest mining when  $q>28.65\%$.
Also we have for all values of $q$ and $\gamma$ that SM2A is superior to SM1. Therefore it is never 
profitable for the attacker to engage in competitions with the honest miners.

\medskip

\begin{figure}[h]
  \resizebox{180pt}{180pt}{\includegraphics{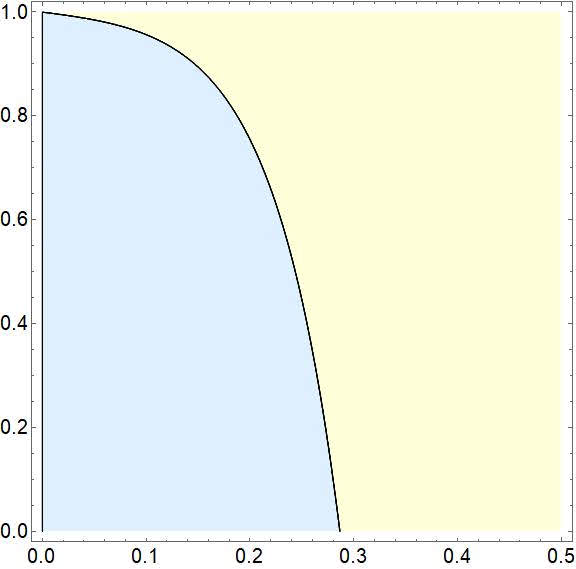}}
  \put(-135, 85){HM}
  \put(-55, 85){SM2A}
  \caption{Comparing HM and SM2A strategies.}
\end{figure}

\subsection{Apparent hashrate of the honest miners}

  We compute first the expected distance between an uncle and its nephew. 
  We keep the same notation for $\Delta$ as in Definition \ref{defDelta}.

  \begin{proposition}
  We have 
  $\mathbb{E} [\Delta] = \frac{q}{p} \left(q (1+p) - \bigl(1+n_1 p\bigr) q^{n_1}\right)$
  \end{proposition}
  
  \begin{proof}
  If an attack cycle $\omega$ starts with S...SH with $k$ S, $k\leq n_1$, 
  then there is exactly one uncle in $\omega$ and its distance to its nephew is $k$.
  In any other cases, there is no uncle in $\omega$. Therefore,
  $\mathbb{E} [\Delta] = \sum_{k=1}^{n_1} k p q^{k}$
  Hence we get the result by (\ref{Del}).
  \end{proof}
  The apparent hashrate $\tilde{p}$ of the honest miners is 
  $\tilde{p} = 1-\tilde{q}-\iota$ with $\iota$ given by (\ref{deflation}). 
  Numerically, we observe that we have always $\tilde{p}< p$ 
  except in a tiny region when $q$ and $\gamma$ is small 
  ($q<6 \%$ and $\gamma < 22 \%$).
  
\medskip

\begin{figure}[h]
  \resizebox{180pt}{180pt}{\includegraphics{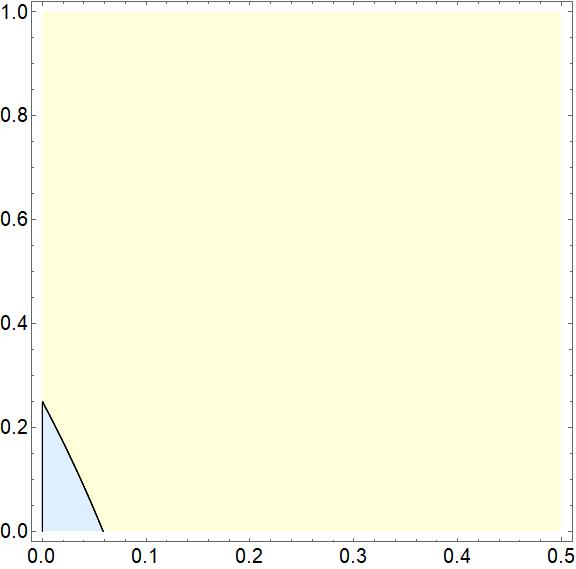}}
  \caption{Comparing $\tilde{p}$ and $p$: 
  The honest miners are negatively affected even when the attack 
  is not profitable for the selfish miner 
  except for a tiny region around $(0,0)$ (case SM2A).}
\end{figure}

\medskip
  
\section{Strategy 2B: Brutal Fork without signaling uncles.}

In this strategy, the attacker signals no uncles in order to 
maximize the impact on the difficulty adjustment formula.
In that case we have $U_s = 0$. In our analysis of the profitability of the strategy, 
we need to consider another important rule of Ethereum's protocol: a nephew can 
only signal at most two uncles. 
Instead of computing $\mathbb{E} [U]$, it is simpler to compute $\mathbb{E} [U']$
where $U' (\omega)$ is defined as the number of signaled uncles with
nephews in $\omega$.
We have,
\begin{equation}\label{eueup}
  \mathbb{E} [U] =\mathbb{E} [U'] 
\end{equation}
Since the attacker does not signal uncles, we have $U' (\omega) = 0$ if
$\omega \notin \{ H, \text{SHH} \}$.

  To ease notations, we set $U'(H)$ for $U'(\{H\})$. 
  
\medskip

\begin{lemma}
  \label{probaU0}
  We have:
  $$
  \begin{array}{rcl}
  \label{piu0n}
  \mathbb{P} [U' (H) = 1] & = & 
  \sum_{i = 2}^{n_1 - 2}
  \big( 1 - pq^2 - pq^2 C_{n_1 - 2 - i} (pq)
  \big)  \pi_{i} + \pi_{n_1 -1} +  \pi_{n_1}\\
  \label{piu0n2}
  \mathbb{P} [U' (H) = 2] &=&
  \sum_{i + j \leq n_1}
  \pi_{i} \pi_{j}\\
  \mathbb{P} [U' (H) \geq 3]&=&0
  \end{array}
  $$

\end{lemma}

\medskip

  \begin{proof}
  We have $U'(H)=1$ if and only if
  the two last attack cycles
  before H are in the following order from the oldest to 
  the most recent one: $\omega'$ and $\omega$ such that:
  \begin{enumerate}
    \item \label{more_condition1_1}
    $\omega$ won by the
    attacker with $L(\omega)\leqslant n_1$.
    \item \label{more_condition1_2}
    $\omega'$ won by the
    honest miners or by the attacker but with 
    $L(\omega') > n_1 - L(\omega)$.
  \end{enumerate}  
  Note that if $L(\omega)\geqslant n_1 - 1$ 
  then  (\ref{more_condition1_2}) is automatically satisfied.
  So, 
  $$
  \mathbb{P} [U' (H) = 1] = 
  \sum_{i = 2}^{n_1 - 2}
  \big( 1 - pq^2 - pq^2 C_{n_1 - 2 - i} (pq)
  \big)  \pi_{i} + \pi_{n_1 -1} +  \pi_{n_1}
  $$
  In the same way, we have $U'(H)=2$ if and only if
  the two last attack cycles
  before H are $\omega'$ and $\omega$ such that 
  $\omega'$ and $\omega$ are both won by the attacker with
  $ L(\omega) +  L(\omega')\leq n_1$. Indeed, a block 
  can only refer at most two uncles. 
  Hence, we get the result.
  \end{proof}

\medskip

\begin{example}\normalfont
  \label{ex06}For $n_1 = 6$, 
  we have using Example \ref{Catalan4}:
  \begin{align*}
    \mathbb{P} [U' (H) = 1] = 
    & \pi_5 + \pi_6 + \sum_{i \leqslant
    4}  (1 - pq^2 - pq^2 C_{4 - i} (pq)) \pi_i\\
    = & 
    p q^2 \left(14 p^4 q^4+p^3 (5-9 q) q^3+2 p^2 (1-2 q) q^2
    +p \left(q-4 q^2\right)+2\right) \\
    \mathbb{P} [U' (H) = 2] = 
    & \pi_2^2 + 2 \pi_2 \pi_3 + 2 \pi_2 \pi_4 + \pi_3^3\\
    = & p^2 q^4 \left(5 p^2 q^2+2 p q+4\right)\\
    \mathbb{P} [U' (H) \geq 3] = &0
  \end{align*}
\end{example}

\medskip

\begin{definition}
  We define $P_{n_1} (p, q) =\mathbb{E} [U' (H)]$.
\end{definition}

\medskip

\begin{example}\normalfont
  When $n_1 = 6$, we have by Example \ref{ex06}:

  \begin{equation*}
  P_6 (p, q) = p q^2 \left(14 p^4 q^4+p^3 (q+5) q^3+2 p^2 q^2+p (4 q+1) q+2\right)
  \end{equation*}

\end{example}

\medskip

\begin{lemma}
  We have 
  $$
  \mathbb{E} [U' (\omega) | \omega = \text{SHH}] =
  (P_{n_1} (p, q) +  1) 
  \cdot (1 - \gamma) 
  + (p q^{2} + p q^{2} C_{n_1 -3} (p q) + 1) 
  \cdot \gamma
  $$
  
\end{lemma}

\medskip

\begin{proof}
  Suppose that $\omega =$SHH.  
  We have two cases: The second honest block can be 
  built on top of a block validated by the selfish miner 
  or not. If the first  official block of $\omega$ is
  honest, then it signals any uncle which is
  at distance less or equal than $n_1$, like in the previous
  situation. Moreover, the first block mined by 
  the selfish miner is an uncle signaled 
  by the second block mined by the honest miners. 
  This gives the first term of the right hand side.
  If the first official block of $\omega$ 
  is a block mined by the attacker, then the first
  block validated by the honest miners is an uncle 
  signaled by the second block mined by the honest miners.
  This last block will also signal another uncle
  which is at distance less than $n_1 - 1$ 
  of the first official
  block of $\omega$. There is such an uncle if and only if
  the attack cycle $\omega'$ before $\omega$ 
  is an attack cycle won by the attacker 
  with $L(\omega')\leq n_1 -1$.  
  This gives the second term 
  of the right hand side. Hence, we get the result.
\end{proof}

\medskip

\begin{theorem}
  We have 
  $$
  \mathbb{E} [U] = (p + (1 - \gamma) p^2 q) P_{n_1} (p, q) 
  + \gamma  p^2 q \left( p q^{2} + p q^{2} C_{n_1 -3} (p q)
  \right) + p^2 q 
  $$
\end{theorem}

\medskip

\begin{proof}
  We have $\mathbb{E} [U] =\mathbb{E} [U']$ and
  \begin{align*}
    \mathbb{E} [U'] & =  \mathbb{E} [U' (\omega) | \omega = H] \mathbb{P}
    [\omega = H] +\mathbb{E} [U' (\omega) | \omega = \text{SHH}] \mathbb{P}
    [\omega = \text{SHH}]\\
    & =  P_{n_1} (p, q) p + (P_{n_1} (p, q) + 1) \cdot (1 - \gamma) p^2 q +
    (p q^{2} + p q^{2} C_{n_1 -3} (p q) + 1) \cdot \gamma p^2 q
  \end{align*}
\end{proof}

\medskip

\subsection{Apparent hashrate of Strategy 2B}

\medskip

The computation of  $\mathbb{E} [U]$ is a polynomial expression in $p$ and $q$ that can be carried out with the help of a computer algebra system.
We plot in parameter space in Figure 3 the region of $(q, \gamma) \in [0, 0.5] \times [0, 1]$ comparing  Selfish Mining Strategies 2A and 2B, and  
honest mining. We also compare SM1, SM2A and SM2B in Figure 4.

\medskip

  We observe that if $\gamma=0$ then we have SM2B is
  superior to honest mining when $q>28.80\%$. Also, for
  $q>30.13\%$ we have that SM2B is even better than SM2A
  (whathever $\gamma$ is). Thus, in this case, the attacker
  does not even need to bother to signal blocks.

\medskip

\begin{figure}[h]
  \resizebox{180pt}{180pt}{\includegraphics{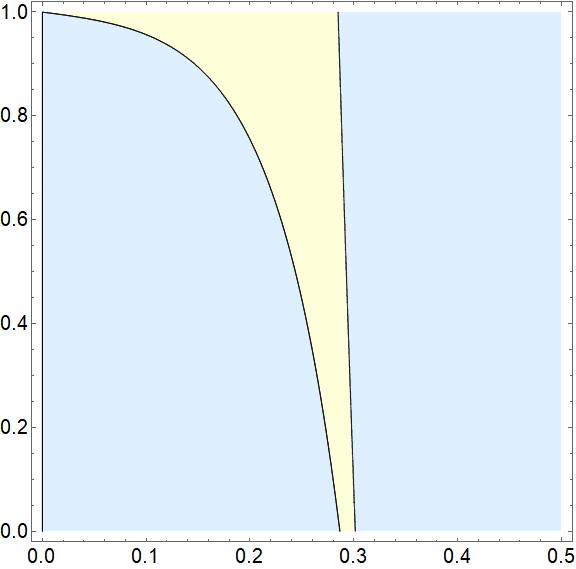}}
   \put(-135, 95){HM}
  \put(-45, 95){SM2B}
  \put(-108, 155){SM2A}
  \caption{Comparing the strategies HM, SM2A and SM2B.}
\end{figure}

\begin{figure}[h]
  \resizebox{180pt}{180pt}{\includegraphics{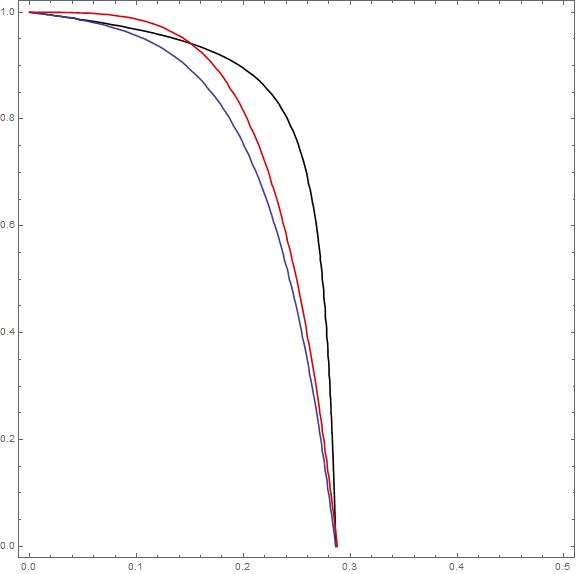}}
  \caption{Comparing the strategies SM1 (black), 
  SM2A (blue) and SM2B (red).}
\end{figure}

\medskip

\subsection{Apparent hashrate of the honest miners}

\medskip

  \begin{definition}
  If $\omega$ is an attack cycle, 
  we denote by $\Delta'(\omega)$ 
  the average number of the distance between 
  a nephew belonging to $\omega$ and an uncle 
  which does not necessarily belong to $\omega$.
  
  \end{definition}
  
  In a similar way as before, we prove:
  \begin{lemma}
  We have 
  
  \begin{equation*}
    \mathbb{E} [\Delta' (\omega) | \omega = H] = 
    \sum_{ |\mathbf{i}|\leqslant n_1}
    \left( \sum_{j} j\cdot i_j \right)
    \big( 1 - pq^2 - pq^2 C_{n_1 - 2 - |\mathbf{i}|} 
    (pq) \big)  \prod_{j}
    \pi_{i_j} 
  \end{equation*}
  \end{lemma}
 
  \begin{definition}
  We define 
  $Q_{n_1}(p,q)=\mathbb{E} [\Delta' (\omega) | \omega = H]$
  \end{definition}
  
 The same computations as in the previous section leads to
 
   \begin{align*}
     Q_5 (p,q) &= p q^2 
     \left(25 p^3 q^3+20 p^2 q^3 
     +8 p^2 q^2+16 p q^2+3 p q+4\right)\\
     Q_6 (p,q) &= p q^2 \left(84 p^4 q^4
     +54 p^3 q^4+25 p^3 q^3+96 p^2 q^4+20 p^2 q^3
     +8 p^2 q^2+16 p q^2+3 p q+4\right)
   \end{align*}
   
   This enables us to compute $\mathbb{E} [\Delta']$ 
   using the following result with $n_1=6$.
   
  $$
  \mathbb{E} [\Delta'] = 
  (p + (1 - \gamma) p^2 q) Q_{n_1} (p, q) + \gamma
  p^2 q Q_{n_1 - 1} (p, q) + p^2 q \ .
  $$
  
  Finally, we note that 
  $\mathbb{E} [\Delta]=\mathbb{E} [\Delta']$. 
  From here, we get the apparent hashrate of 
  the honest miners using (\ref{deflation}) and
  (\ref{pqtilde}). We observe numerically that
   we have always  ${\tilde p} < p$.

\medskip

\section{Conclusions}

\medskip

We have given closed-form formulas for the long term profitability of different selfish mining strategies in the 
Ethereum network. This is combinatorially more complex than in Bitcoin network which has a simpler reward system. 
Precisely, the particular reward system that incentives signaling blocks is an effective counter-measure 
to Selfish mining but only when the count of uncle blocks are incorporated into the difficulty adjustment formula 
(this is the case for the current implementation of the difficulty adjustment formula). 
This analysis provides a good illustration of the fact that selfish mining is an attack on the difficulty adjustment formula.
We study, for the first time, selfish mining strategies that do not signal any blocks. We prove that they are the most 
profitable ones in the long run. It may appear counter-intuitive that refusing the signaling fees is the most profitable 
strategy with the current reward parameters when $q$ is larger than $30\%$. But this is explained again because selfish mining 
is an attack on the difficulty adjustment formula.

\section*{Appendix}

\subsection{Random walk}

We compute the expected numbers of descents in a biased random walk conditional to be bounded by a fixed bound.

\medskip

\begin{lemma}
  \label{maon}Let $(X_k)$ be a biased random walk starting from $X_0 = 0$ with
  $\mathbb{P} [X_{k + 1} = X_k + 1] = q$ and $\mathbb{P} [X_{k + 1} = X_k -
  1] = p$ for $k \in \mathbb{N}$, with $p + q = 1$ and $q < p$. 
  Let $\nu (X)$ be the stopping time defined by 
  $\nu (X) = \inf \{ i \geq 0; X_i= - 1 \}$, 
  and for $n \geq 0$, let 
  \begin{align*}
  u_n (X) =& 
  \sum_{i=1}^{\nu}  {\bf 1}_{(X_i < n) \wedge (X_i < X_{i - 1})}\\
  v_n (X) =& 
  \sum_{i=1}^{\nu}  X_i \cdot {\bf 1}_{(X_i < n) \wedge (X_i < X_{i - 1})}
  \end{align*} 
  Then we have
  \begin{align}
  u_n = \mathbb{E} [u_n (X)] =& \frac{p}{p - q}  \left( 1 -
  \left( \frac{q}{p}^{} \right)^{n + 1} \right)\label{equn}\\
  v_n = \mathbb{E} [v_n (X)] =& 
  \frac{p}{(p-q)^2} \left( 2q-p - \bigl(q+n (p-q)\bigr)\cdot \left( \frac{q}{p}\right)^{n+1} \right) \label{eqvn}
  \end{align}
\end{lemma}

\medskip

\begin{proof}
  We have $u_0 = 1$ (resp. $v_0=-1$). 
  If $X_1 = - 1$, then we have $u_n(X)=1$ (resp. $v_n(X)=-1$). 
  If $X_1 = 1$, then 
    \begin{align*}
    u_n(X) 
    =& \sum_{i=1}^{\nu'}  {\bf 1}_{(X'_i < n-1) \wedge (X'_i < X'_{i - 1})}
     + \sum_{i=1}^{\nu^{''}}  {\bf 1}_{(X^{''}_i < n) \wedge (X^{''}_i < X^{''}_{i - 1})} \\
    =& u_{n-1}(X')+u_n(X^{''})
    \end{align*}

  with

    \begin{align*}
    X'_i =& X_{i + 1} - 1 \\
    \nu'=& \inf \{ i >0 ; X'_i = -1 \} \\
    X_i^{''} =& X'_{i + \nu'} - X'_{\nu'} \\
    \nu^{''} =& \inf \{ i >0 ; X_i^{''} = -1 \}
    \end{align*}
  By the Markov property, $X'$ and $X^{''}$ are two independant simple biased random walk 
  with a probability $p$ (resp. $q$) to move to the left (resp. right).
  So, taking expectations, we get:
  $$ 
  u_n = p \cdot 1 + q \cdot (u_{n - 1} + u_n) 
  $$
  which is equivalent to
  $$ u_n - \frac{p}{p - q} = \left( \frac{q}{p} \right)  \left(
     u_{n - 1} - \frac{p}{p - q} \right)
  $$
  So we get (\ref{equn}) by induction on $n$.
  In the same way, we have:
    \begin{align*}
    v_n(X) =& \sum_{i=1}^{\nu'}  (X'_i+1) \cdot {\bf 1}_{(X'_i < n-1) \wedge (X'_i < X'_{i - 1})}
              + \sum_{i=1}^{\nu"} X^{''}_i \cdot {\bf 1}_{(X^{''}_i < n) \wedge (X^{''}_i < X^{''}_{i - 1})} \\
    =& u_{n-1}(X')+v_{n-1}(X')+v_n(X^{''})
    \end{align*}
  Taking expectations again, we get
  \begin{equation}
  \label{Eqvn}
  v_n = p \cdot (-1) + q \cdot (u_{n - 1}+v_{n-1} + v_n) 
  \end{equation}
  Set $c_n = \left( \frac{p}{q} \right)^{n} v_n$. Then, (\ref{Eqvn}) leads to
  \begin{align*}
  c_n =& c_{n-1} + \left( \frac{p}{q} \right)^{n-1} u_{n-1} - \left( \frac{p}{q} \right)^{n} \\
  =& c_{n-1} + \left( \frac{2q-p}{p-q}\right) \cdot \left( \frac{p}{q} \right)^{n} - \frac{q}{p-q}
  \end{align*}
  So, by induction, we get
  $$
  c_n = c_0 + \left( \frac{2q-p}{p-q}\right) \cdot \left( \frac{p}{q} \right)\cdot
  \frac{\bigl(\frac{p}{q}\bigr)^{n}-1}{\bigl(\frac{p}{q}\bigr) - 1} - \frac{n q}{p-q}
  $$
  After rearranging terms, we get (\ref{eqvn}).
\end{proof}

\medskip

\subsection{Dyck words}

Let $\mathcal{D}$ be the space of Dyck words based on the alphabet $\{ S, H \}$.
If $w = w_1 \ldots w_{2 k}$ with $k \in \mathbb{N}$, then we define $| w | = k$.
We have proved in {\cite{GPM19a}} that we can endow $\mathcal{D}$ with a
probability measure $\bar{\mathbb{P}}$ given by $\bar{\mathbb{P}} [w] = p
(pq)^{| w |}$ for $w \in \mathcal{D}$. Note that $\bar{\mathbb{P}} [w]$ can
be interpreted as the probability that a simple biased random walk $X$
starting from 0 and stopping at $- 1$ follows exactly the path given by $w$
i.e., $X_i = X_{i - 1} + 1$ (resp. $X_i = X_{i - 1} - 1$) if $w_i = S$ (resp.
$w_i = H$) for $i \leqslant 2 | w |$ and $X_{2 | w | + 1} = - 1$.

\medskip

\begin{lemma}
  \label{pbdn}Let $n \geq 0$ and $\mathcal{D}_n = \{ w ; | w |
  \leqslant n \}$. Then, $\bar{\mathbb{P}} [\mathcal{D}_n] = pC_n (pq)$ where
  $C_n (x)$ is 
  the $n$-th partial sum of the
  generating series $C (x)$ of the Catalan numbers.
\end{lemma}

\medskip

\begin{proof}
  We have
  $$
  \bar{\mathbb{P}} [\mathcal{D}_n] = \underset{w \in
  \mathcal{D}_n}{\sum} p (pq)^{| w |} = p \underset{k = 0}{\overset{n}{\sum}} 
  \underset{| w | = k}{\sum} (pq)^k = p \underset{k = 0}{\overset{n}{\sum}}
  C_k  (pq)^k = pC_n (pq)
  $$
\end{proof}

\medskip

We can make more precise Proposition \ref{bcdyck}.


\begin{proposition}
  \label{ppbar}Let $\omega = SSwH$ be an attack cycle starting with SS. 
  Then, $w\in \mathcal{D}$ and $\mathbb{P}[\omega] = q^{2}\bar{\mathbb{P}} [w]$ 
\end{proposition}

\medskip

\begin{lemma}
  \label{pdyck}The probability that a Dyck word ends with the subsequence
  SHH..H with $n$ letters H at the end is $pq^n$.
\end{lemma}


\begin{proof}
  Consider the ``reversal'' map $\sigma : \quad \mathcal{D}  \longrightarrow  \mathcal{D}$ given by
  \begin{equation*}
    w = w_1 \ldots w_{2 | w |}  \longmapsto  \sigma(w)=\tilde{w} = \tilde{w}_{2 | w |}
    \ldots \tilde{w}_1
  \end{equation*}
  with $\tilde{w}_i = S$ (resp. $H$) if $w_i = H$ (resp. $S$). Then $\sigma$ is one to one and preserves $\bar{\mathbb{P}}$ i.e.,
  for $w \in \mathcal{D}$, we have $\bar{\mathbb{P}} [\sigma (w)] =
  \bar{\mathbb{P}} [w]$. So, the probability that a Dyck word ends exactly
  with $n$ letter(s) H is the same as the probability that a Dyck word starts
  with $n$ letter(s) S and then is followed by a letter H. Thus this probability
  is $pq^n$.
\end{proof}

\medskip

For $w \in \mathcal{D}$, we define $f (w) = \inf \{ i \geq 0; w_i
= H \}$.


\begin{lemma}\label{pbare}
  Let $n \geq 0$ and $E = \{ w \in \mathcal{D}; f (w)
  \leqslant \inf \{ | w |, n \} \}$. Then we have 
  $$
  \bar{\mathbb{P}} [E] = (1 -q^n) - \frac{p (1 - (pq)^n)}{1 - pq}
  $$
\end{lemma}


\begin{proof}
  Let $w\in \mathcal{D}$. To have $f (w) \leqslant | w |$
  means that at least one H is followed by an S i.e., $w$ is not of the form
  SS...SHH...H. For all integer $k \leqslant n$, we have 
  $$
  \Sigma_{w ; (f (w) =k) \wedge (f (w) \leqslant | w |)}  (pq)^{| w |} = pq^{k - 1} \cdot \sum_{j
  = 0}^{k - 2} qp^j
  $$ 
  So, if we consider a biased random walk starting from $0$
  with a probability $p$ to move to the left (resp. right) then both terms
  represent the probability of the following event: We have $k - 1$ first 
  step(s) to the right, then $j + 1$ steps to the left with $0 \leqslant j
  \leqslant k - 2$ and then at least one step to the right before reaching
  $0$. So, we have
  \begin{align*}
    \bar{\mathbb{P}} [E] & =  \sum_{k = 1}^n pq^{k - 1} \cdot \sum_{j =
    0}^{k - 2} qp^j\\
    & =  p \sum_{k = 1}^n q^{k - 1} \cdot (1 - p^{k - 1})\\
    & =  p \sum_{k = 1}^n q^{k - 1} - p \sum_{k = 1}^n (pq)^{k - 1}\\
    & = (1 - q^n) - \frac{p (1 - (pq)^n)}{1 - pq}
  \end{align*}
\end{proof}


\subsection{Glossary}

\subsubsection{Revenue ratio and apparent hashrate}
  The revenue ratio ${\tilde{\Gamma}}$ of a miner following 
  a strategy with repetitions of attack cycles like selfish mining 
  is given by ${\tilde{\Gamma}} = \frac{\mathbb{E} [R]}{\mathbb{E} [T]}$
  where $R$ (resp. $T$) is the revenue of the miner after an attack 
  cycle  (resp. the duration time of an attack cycle). The apparent
  hashrate ${\tilde{q}}$ is defined by
  ${\tilde{q}} = {\tilde{\Gamma}} \frac{\tau}{b}$ 
  where $b$ (resp. $\tau$) is the coinbase (resp. interblock time). 

\subsubsection{Terminology}

Ethereum has a special terminology that we summarize.

\

\begin{tabular}{|l|l|}
  \hline
  Uncle & orphan block whose parent belongs to the official blockchain\\
  \hline
  Nephew & regular block that refers to an ``uncle'' which is at a distance
  less than $n_1$\\
  \hline
  Distance & number of official blocks between a nephew N and a parent's uncle
  U.\\
  \hline
\end{tabular}

\subsubsection{Mining reward}

If an uncle U is referred by a nephew N which is at a distance $d$, then U earns an
``uncle reward'' which is worth $K_u (d) b$ and N gets an additional reward of
$K_n (d) b$, where $b$ is the coinbase. Also, a nephew can refer at most two uncles.
Today, on Ethereum, we have $b = 2 \text{ ETH}$,
$K_u (d) = \frac{8 - d}{8} \cdot {\bf{1}}_{d \leqslant n_1}$ with $n_1 =
6$ and $K_n (d) = \pi = \frac{1}{32}$.

\medskip

\begin{tabular}{|l|l|}
  \hline
  Uncle reward & reward granted to an uncle block referred by a nephew \\
  \hline
  inclusion reward & additional reward granted to a nephew that refers an
  uncle\\
  \hline
\end{tabular}


\end{document}